\documentclass[conference]{IEEEtran}
\ifCLASSINFOpdf
  \usepackage[pdftex]{graphicx}
  \graphicspath{{/Users/zahmed/mySharedFolder/per/PhD/Latex/IEEEtran/pdf/}{/Users/zahmed/mySharedFolder/per/PhD/Latex/IEEEtran/jpeg/}}
  \DeclareGraphicsExtensions{.pdf,.jpeg,.png}
\else
  \DeclareGraphicsExtensions{.eps}
\fi
\usepackage{algorithm}
\usepackage{algpseudocode}
\usepackage{tabu}
\usepackage{setspace}
\usepackage{makecell}
\usepackage{epsfig}
\usepackage[usenames, dvipsnames]{color}
\usepackage[roman]{parnotes}
\usepackage{amsthm,amsmath,amssymb}
\usepackage{mathtools}
\usepackage{breqn}
\usepackage{amsfonts}
\usepackage{amssymb}
\usepackage{mathrsfs}
\usepackage{xparse}
\usepackage{multirow}
\usepackage{chngcntr}

\DeclarePairedDelimiter{\diagfencesbigg}{\Bigg(}{\Bigg)}
\newcommand{\diagBigg}{\operatorname{diag}\diagfencesbigg}
\DeclarePairedDelimiter{\diagfences}{(}{)}
\newcommand{\diag}{\operatorname{diag}\diagfences}
\allowdisplaybreaks
\begin{document}
%
\title{Capacity analysis and bit allocation design for variable-resolution ADCs in Massive MIMO}
\author{\IEEEauthorblockN{I. Zakir Ahmed  and Hamid Sadjadpour\\}
\IEEEauthorblockA{Department of Electrical Engineering\\
University of California, Santa Cruz\\
}
\and
\IEEEauthorblockN{Shahram Yousefi}
\IEEEauthorblockA{Department of Electrical and Computer Engineering\\
Queen's University, Canada\\
}}


%


\maketitle

\begin{abstract}
We derive an expression for the capacity of massive multiple-input multiple-output Millimeter wave (mmWave) channel where the receiver is equipped with a variable-resolution Analog to Digital Converter (ADC) and a hybrid combiner. 
The capacity is shown to be 
a function of Cramer-Rao Lower Bound (CRLB) for a given bit-allocation matrix and hybrid combiner. 
The condition for optimal ADC bit-allocation under a receiver power constraint is derived. This is derived based on the maximization of capacity with respect to bit-allocation matrix for a given channel, hybrid precoder, and hybrid combiner.
It is shown that this condition coincides with that obtained using the CRLB minimization proposed by Ahmed et al. 
 Monte-carlo simulations show that the capacity calculated using the proposed condition matches very closely with the capacity obtained using the Exhaustive Search bit allocation.
\end{abstract}


%

%
\IEEEpeerreviewmaketitle

\newcommand{\Xmatrix}{
\begin{bmatrix}
\ddots  & 0     & 0 \\
0  & \frac{1}{{\sigma_i^2}} & 0 \\
0 & 0 & \ddots
\end{bmatrix}}
\newcommand{\Ymatrix}{
\begin{bmatrix}
\ddots  & 0  & 0 \\
0  & \frac{f(b_i)l_i}{\big(1-f(b_i)\big) \sigma_i^2} & 0 \\
0 & 0 & \ddots
\end{bmatrix}}
\newcommand{\Zmatrix}{
\begin{bmatrix}
\ddots  & 0  & 0 \\
0  & \frac{\sigma_i^2}{\sigma_n^2 + \frac{f(b_i)l_i}{\big(1-f(b_i)\big)}} & 0 \\
0 & 0 & \ddots
\end{bmatrix}}
\newcommand{\ZMmatrix}{
\begin{bmatrix}
\ddots  & 0  & 0 \\
0  & \frac{\sigma_i^2}{\sigma_n^2 + \frac{f(b_i)l_i}{\big(1-f(b_i)\big)}} + \frac{1}{p} & 0 \\
0 & 0 & \ddots
\end{bmatrix}}
\newcommand{\Fmatrix}{
\begin{bmatrix}
\ddots  & 0  & 0 \\
0  & 10 & 0 \\
0 & 0 & \ddots 
\end{bmatrix}}
\newcommand{\InvFmatrix}{
\begin{bmatrix}
\ddots  & 0 & 0 \\
0  & f(b_i)\big(1-f(b_i)\big)l_i & 0 \\
0 & 0 & \ddots
\end{bmatrix}}

\section{Introduction}
Massive Multiple-Input Multiple-Output (MIMO) is a key feature for next generation of wireless communication standards. It is being considered both at sub-6Ghz frequencies and at mmWave frequencies \cite{5GBackHaul,SigProc}. In both scenarios, a large number of antennas help increase the capacity of the system through spatial multiplexing or increasing the energy efficiency by focusing on the intended user through beam-forming. With favorable channel conditions, a combination of both can be achieved. Hybrid precoding and combining can be used to capture the potential advantages. As such, precoders and combiners both in analog and digital domains, are adapted based on the changing channel conditions \cite{SigProc,mmPreCom}.

However, increasing the number of RF paths increases the cost of RF components and the power consumption. One of the power-hungry components in massive MIMO receivers operating at larger bandwidths and high resolution is the Analog to Digital Converter (ADC) \cite{Rangan}. In addition to power consumption, high resolution ADCs operating at high sampling frequencies produce huge amount of data that is difficult to handle. 

 \subsection{Previous Works} \label{pwork}
In previous works, the usage of low-medium-bit-resolution ADCs in massive MIMO has shown promising results in terms of energy efficiency and cost \cite{Muris,HybArchCap}. Analysis of low resolution and in particular 1-bit ADCs for massive MIMO architectures has received a great deal of attention \cite{SigProc,mmPreCom,Mezghani,HybArchCap}. Capacity and performance analysis with 1-bit ADCs were discussed in \cite{Mezghani}. Capacity analysis with 1-bit ADC using precoders are studied in \cite{Jmo}. Achievable rate analysis of 1-bit ADCs together with hybrid combining were discussed and evaluated in \cite{HybArchCap}.\\
\indent
Analysis of uniform low-resolution ADCs are considered in \cite{Rangan,Muris,HybArchCap,Uplink}. Approximate rate expression for low-resolution $n$-bit ADCs on all RF paths considering quantization as an Additive Quantization Noise Model (AQNM) was derived in \cite{Uplink}. The effect of low-resolution ADCs and bandwidth on the achievable rates using AQNM under a receiver power constraint was investigated in \cite{Rangan}. A generalized hybrid combiner with low-resolution ADCs was studied in \cite{HybArchCap}.\\ 
\indent
Adopting variable-bit resolution ADCs across RF paths is shown to improve the Mean Squared Error (MSE) performance of the receiver under some channel conditions  \cite{Zakir1}. A bit-allocation strategy is needed to achieve this MSE performance, under a power constraint. A near-optimal low complexity bit-allocation algorithm under a power constraint was presented in \cite{VarBitAlloc}. A joint combiner design and bit allocation framework using genetic algorithm to minimize the MSE of the quantized and combined symbol is devised in \cite{Zakir1}. An optimal condition for ADC bit-allocation for mmWave massive MIMO is derived by minimizing the Cramer-Rao Lower Bound (CRLB) which in effect minimizes the MSE of the quantized and combined symbol under a power constraint in \cite{Zakir2}.

 \subsection{Our Contribution}
 In this work, we derive a condition for optimal ADC bit-allocation based on maximizing the capacity for a given channel under a  receiver power constraint.
 The two main contributions of this paper are as follows:
 
 \textit{i)} We derive the capacity expression for a given mmWave channel as a function of bit-allocation matrix and hybrid combining matrices. The bit-allocation matrix facilitates variable-bit allocation on the receiver's RF paths. In addition, we show that this capacity is a function of the CRLB. It is shown that there exists a Minimum-Mean-Squared-Error (MMSE) estimator for the transmitted symbol vector that achieves this CRLB \cite{Zakir2}.
 
 \textit{ii)} We design an ADC bit-allocation algorithm for a given power budget based on maximizing the capacity of a given channel and hybrid combiners. In doing so, we arrive at exactly the same conditions and algorithm that was previously found under the minimization of CRLB \cite{Zakir2}. We substantiate our results with simulations.
\subsection{Notation}
We represent the column vectors using boldface small letters,  matrices as boldface uppercase letters,  the primary diagonal of a matrix  as $\text{diag}(\cdot)$, and the Expectation $E[\cdot]$ is over the random variable $\bold{n}$, which is an AWGN vector. The multivariate Gaussian distribution with mean $\boldsymbol{\mu}$ and covariance $\boldsymbol{\varphi}$ is denoted as $\mathcal{N}(\boldsymbol{\mu},\boldsymbol{\varphi})$ and the multivariate complex-valued circularly-symmetric Gaussian distribution with zero mean and covariance $\boldsymbol{\varphi}$ is denoted as $\mathcal{CN}(\bold{0},{\boldsymbol{\varphi}})$. The superscripts $T$ and $H$ denote transpose and Hermitian transpose, respectively.

\section{Signal Model}
The signal model for a transceiver encompassing hybrid precoding and combining for a mmWave MIMO channel is shown in Figure \ref{fig:Fig1new.pdf}. Let ${\bold{F}_D}$ and ${\bold{F}_A}$ be the digital and analog precoders, respectively, and ${\bold{W}_D^H}$ and ${\bold{W}_A^H}$ the digital and analog combiners. Also, $\bold{H} = \big[ h_{ij} \big]$ represents the \begin{math}N_r\times N_t\end{math} line of sight mmWave MIMO channel with properties defined in \cite{rapaport}(chapter 3, pages 99-125). We define $N_t$ and $N_r$ to represent the number of transmit and receive antennas, respectively. Let $\text{Q}_{\bold{b}} \big( {\bold{z}} \big)$ represent the AQNM as defined in \cite{VarBitAlloc}. Let $\bold{n}_q$ be the additive quantization noise vector uncorrelated with $\bold{z}$ Gaussian distributed as ${\bold{n}_q} \sim \mathcal{CN}(\bold{0},{\bold{D}_q^2})$, where ${\bold{D}_q^2} = {\bold{W}_{\alpha}}{\bold{W}_{1-\alpha}}{\text{diag}}[ {\bold{W}_A^H}{\bold{H}}({\bold{W}_A^H}{\bold{H}})^H+{\bold{I}_{N_{rs}}}]$ and 
$\bold{W}_{\alpha}\big( {\bold{b}} \big)$ is the diagonal bit-allocation matrix $\cite{VarBitAlloc}$. 
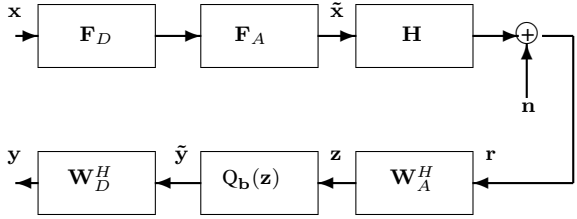
\begin{figure}[h]
\begin{center}
\setlength{\unitlength}{3947sp}%
\begingroup\makeatletter\ifx\SetFigFont\undefined%
\gdef\SetFigFont#1#2#3#4#5{%
  \reset@font\fontsize{#1}{#2pt}%
  \fontfamily{#3}\fontseries{#4}\fontshape{#5}%
  \selectfont}%
\fi\endgroup%
\begin{picture}(3589,1349)(136,-861)
\put(3151,-511){\makebox(0,0)[lb]{\smash{{\SetFigFont{8}{9.6}{\familydefault}{\mddefault}{\updefault}{\color[rgb]{0,0,0}$\bold{r}$}%
}}}}
\thinlines
{\color[rgb]{0,0,0}\put(345, 75){\framebox(730,389){}}
}%
{\color[rgb]{0,0,0}\put(2341, 75){\framebox(729,389){}}
}%
{\color[rgb]{0,0,0}\put(2341,-849){\framebox(729,389){}}
}%
{\color[rgb]{0,0,0}\put(1367,-849){\framebox(730,389){}}
}%
{\color[rgb]{0,0,0}\put(345,-849){\framebox(730,389){}}
}%
\thicklines
{\color[rgb]{0,0,0}\put(1075,270){\vector( 1, 0){292}}
}%
{\color[rgb]{0,0,0}\put(2097,270){\vector( 1, 0){244}}
}%
{\color[rgb]{0,0,0}\put(3070,270){\vector( 1, 0){291}}
}%
{\color[rgb]{0,0,0}\put(3411,-120){\vector( 0, 1){341}}
}%
{\color[rgb]{0,0,0}\put(3507,270){\line( 1, 0){196}}
\put(3703,270){\line( 0,-1){925}}
\put(3703,-655){\vector(-1, 0){633}}
}%
{\color[rgb]{0,0,0}\put(2341,-655){\vector(-1, 0){244}}
}%
{\color[rgb]{0,0,0}\put(1367,-655){\vector(-1, 0){292}}
}%
{\color[rgb]{0,0,0}\put(199,270){\vector( 1, 0){146}}
}%
{\color[rgb]{0,0,0}\put(345,-655){\vector(-1, 0){146}}
}%
\thinlines
{\color[rgb]{0,0,0}\put(1367, 75){\framebox(730,389){}}
}%
\put(2551,-661){\makebox(0,0)[lb]{\smash{{\SetFigFont{8}{9.6}{\familydefault}{\mddefault}{\updefault}{\color[rgb]{0,0,0}$\bold{W}_A^H$}%
}}}}
\put(1501,-661){\makebox(0,0)[lb]{\smash{{\SetFigFont{8}{9.6}{\familydefault}{\mddefault}{\updefault}{\color[rgb]{0,0,0}$\text{Q}_{\bold{b}} \big( {\bold{z}} \big)$}%
}}}}
\put(526,-661){\makebox(0,0)[lb]{\smash{{\SetFigFont{8}{9.6}{\familydefault}{\mddefault}{\updefault}{\color[rgb]{0,0,0}$\bold{W}_D^H$}%
}}}}
\put(3376,-211){\makebox(0,0)[lb]{\smash{{\SetFigFont{8}{9.6}{\familydefault}{\mddefault}{\updefault}{\color[rgb]{0,0,0}$\bold{n}$}%
}}}}
\put(2176,389){\makebox(0,0)[lb]{\smash{{\SetFigFont{8}{9.6}{\familydefault}{\mddefault}{\updefault}{\color[rgb]{0,0,0}$\bold{\tilde{x}}$}%
}}}}
\put(151,389){\makebox(0,0)[lb]{\smash{{\SetFigFont{8}{9.6}{\familydefault}{\mddefault}{\updefault}{\color[rgb]{0,0,0}$\bold{x}$}%
}}}}
\put(2626,239){\makebox(0,0)[lb]{\smash{{\SetFigFont{8}{9.6}{\familydefault}{\mddefault}{\updefault}{\color[rgb]{0,0,0}$\bold{H}$}%
}}}}
\put(1576,239){\makebox(0,0)[lb]{\smash{{\SetFigFont{8}{9.6}{\familydefault}{\mddefault}{\updefault}{\color[rgb]{0,0,0}$\bold{F}_A$}%
}}}}
\put(601,239){\makebox(0,0)[lb]{\smash{{\SetFigFont{8}{9.6}{\familydefault}{\mddefault}{\updefault}{\color[rgb]{0,0,0}$\bold{F}_D$}%
}}}}
\put(1201,-511){\makebox(0,0)[lb]{\smash{{\SetFigFont{8}{9.6}{\familydefault}{\mddefault}{\updefault}{\color[rgb]{0,0,0}$\bold{\tilde{y}}$}%
}}}}
\put(2176,-511){\makebox(0,0)[lb]{\smash{{\SetFigFont{8}{9.6}{\familydefault}{\mddefault}{\updefault}{\color[rgb]{0,0,0}$\bold{z}$}%
}}}}
\put(151,-511){\makebox(0,0)[lb]{\smash{{\SetFigFont{8}{9.6}{\familydefault}{\mddefault}{\updefault}{\color[rgb]{0,0,0}$\bold{y}$}%
}}}}
\put(3376,239){\makebox(0,0)[lb]{\smash{{\SetFigFont{9}{10.8}{\rmdefault}{\mddefault}{\updefault}{\color[rgb]{0,0,0}+}%
}}}}
{\color[rgb]{0,0,0}\put(3411,288){\circle{136}}
}%
\end{picture}%
\caption{Signal Model}
\label{fig:Fig1new.pdf}
\end{center}
\end{figure}

The transmitted symbol $\bold{x}$ is a vector of size $Ns\times1$ whose average power is $p$. Let ${\bold{n}}$ be a $N_r\times1$ noise vector of independent and identically distributed (i.i.d.) complex Gaussian random variables such that ${\bold{n}} \sim \mathcal{CN}(\bold{0},{\sigma_n^2}{\bold{I}_{N_r}})$. 

We assume that we have perfect Channel State Information (CSI) at the transmitter and the number of RF paths $N_{rs}$ at the receiver is the same as the number of parallel data streams $N_s$, i.e., $N_{rs} = N_s$. The analysis can be easily extended to the case $N_{rs} \ne N_s$.
 
The  dimensions of matrices indicated in Figure $\ref{fig:Fig1new.pdf}$ are as follows: 
${\bold{F}_D} \in \mathbb{C}^{N_{rt} \times N_s}$, ${\bold{F}_A} \in \mathbb{C}^{N_t \times N_{rt}}$, ${\bold{H}} \in \mathbb{C}^{N_r \times N_t}$, ${\bold{W}_A^H} \in \mathbb{C}^{N_{rs} \times N_r}$, ${\bold{W}_D^H} \in \mathbb{C}^{N_s \times N_{rs}}$, $\bold{W}_{\alpha}\big( {\bold{b}} \big) \in \mathbb{R}^{N_{rs} \times N_{rs}}$.
\subsection{Precoders and Combiners design}
We assume that the hybrid precoder is designed independent to that of combiners or bit-allocation. One could use the technique proposed in $\cite{SigProc,PreDsgn}$ for the same.\\
\indent
The hybrid combiner is designed based on the Singular Value Decomposition (SVD) of the given channel matrix \cite{Zakir2} as
\begin{equation}\label{CombDesign}
\begin{split}
{\bold{W}_A^H} = &{\bold{U}^H} = {\bold{W}_D}{\bold{\tilde{W}}_A^H},\\
\text{or } &{\bold{U}} = {\bold{\tilde{W}}_A}{\bold{W}_D^H}.\\
\end{split}
\end{equation}
Here, ${\bold{\tilde{W}}_A^H}$ is the actual analog combiner that factors in the constraints imposed by the phase shifters or splitters $\cite{SigProc}$. The imperfections in the analog combiner are compensated by the digital combiner ${\bold{W}_D}^H$. Also, ${\bold{U}} \in \mathbb{C}^{N_r \times N_s}$ is the left singular matrix of the SVD of the channel matrix $\bold{H}$ as $\bold{H} = \bold{U}\bold{\Sigma}\bold{F}_{\text{opt}}^H$ where
${\bold{U}} \in \mathbb{C}^{N_r \times N_s}, {\bold{\Sigma}} \in \mathbb{R}^{N_s \times N_s}, \text{ and } {\bold{F}_{\text{opt}}} \in \mathbb{C}^{N_t \times N_s}.$
The relationship between the transmitted signal vector $\bold{x}$ and the received symbol vector $\bold{y}$ is given by
\begin{equation}\label{eq5a}
\begin{split}
{\bold{y}} &= {\bold{W}_D^H}{\bold{W}_{\alpha}}{\big( {\bold{b}} \big)}{\bold{W}_A^H}{\bold{H}}{\bold{F}_A}{\bold{F}_D}{\bold{x}} + {\bold{W}_D^H}{\bold{W}_{\alpha}\big( {\bold{b}} \big)}{\bold{W}_A^H}{\bold{n}} \\
&+{\bold{W}_D^H}{\bold{n_q}}.
\end{split}
\end{equation}
Equation \eqref{eq5a} can be simplified as 
\begin{equation}\label{eq9a}
\begin{aligned}
\begin{split}
{\bold{y}} &= {\bold{W}_D^H}{\bold{W}_{\alpha}}{\big( {\bold{b}} \big)}{\bold{W}_A^H}{\bold{U}}{\bold{\Sigma}}{\bold{x}} + {\bold{W}_D^H}{\bold{W}_{\alpha}\big( {\bold{b}} \big)}{\bold{W}_A^H}{\bold{n}} \\
&+{\bold{W}_D^H}{\bold{n_q}},\\
{\bold{y}} &= {\bold{K}}{\bold{x}} + {\bold{n_1}}
\end{split}
\end{aligned}
\end{equation}
where 
${\bold{K}} = {\bold{W}_D^H}{\bold{W}_{\alpha}}{\bold{W}_A^H}{\bold{U}}{\bold{\Sigma}}$, and 
$\bold{n_1} = {\bold{W}_D^H}{\bold{W}_{\alpha}}{\bold{W}_A^H}{\bold{n}} + {\bold{W}_D^H}{\bold{n_q}}.$
Since
$E[{\bold{x}}{\bold{x}}^H] = p{\bold{I}_{N_s}}$, ${\bold{G}} = {\bold{W}_D^H}{\bold{W}_{\alpha}}{\bold{W}_A^H}$, $E[{\bold{n}}{\bold{n}}^H] = {\sigma_n^2}{\bold{I}_{N_r}}$, $E[{\bold{n_q}}{\bold{n_q}}^H] = {\bold{D}_q^2}$, 
where ${\bold{D}_q^2} = {\bold{W}_{\alpha}}{\bold{W}_{1-\alpha}}{\text{diag}}[ {\bold{W}_A^H}{\bold{H}}({\bold{W}_A^H}{\bold{H}})^H+{\bold{I}_{N_{rs}}}]$, and $E[{\bold{n}}{\bold{n_q}}^H] = 0,$ then 
the statistical distribution of $\bold{n_1}$ is $\bold{n_1} \sim \mathcal{N}(\bold{0},\bold{\Phi})$, where $\bold{\Phi} = {\sigma_n^2}{\bold{G}}{\bold{G}^H} + {\bold{W}_D^H}{\bold{D}_q^2}{\bold{W}_D}$ \cite{Zakir2}. For simplicity of notation we will refer to $\bold{W}_{\alpha}\big( {\bold{b}} \big)$ as $\bold{W}_{\alpha}$.
\section{Capacity Analysis}
The instantaneous capacity for a given MIMO channel with ADC power constraint and bit allocation can be written as
\begin{equation}\label{eq14a}
\begin{aligned}
C =  \{ \underbrace{\text{max}}_{\bold{b} \in \mathbb{I}^{N_s \times 1};{P_{\text{TOT}}}\leq{P_{\text{ADC}}}}{ I\big(\bold{x}; \bold{y}\big) }, \}
\end{aligned}
\end{equation}
where $\bold{b}=[b_1 b_2 b_3 .... b_N]^T$ is a vector whose entries $b_i$ indicate the number of bits $b_i$ (on both I and Q channels) that are allocated to the ADC on RF path $i$. $P_{\text{TOT}}$ is the total power consumed by the ADCs and is known to be $P_{\text{TOT}} = \sum_{i=1}^{N} c{f_s}2^{b_i}$,  where $c$ is the power consumed per conversion step and $f_s$ is the sampling rate in Hz $\cite{Uplink}$. $P_{\text{ADC}}$ is the allowed ADC power budget. Note that the mutual information in ($\ref{eq14a}$) is maximized with respect to the bit-allocation matrix ${\bold{W}_{\alpha}}{\big( {\bold{b}} \big)}$. Equation \eqref{eq14a} can be written  \cite{Thomas} as
\begin{equation}\label{eq15a}
\begin{split}
I(\bold{x};\bold{y}) &= h(\bold{y}) - h(\bold{y}|\bold{x}),\\
&= h(\bold{y}) - h(\bold{K}\bold{x} + \bold{n_1}|\bold{x}),\\
&= h(\bold{y}) - h(\bold{n_1}),
\end{split}
\end{equation}
where $h(\cdot)$ is the differential entropy of a continuous random variable. We assume that $\bold{x}$ and $\bold{n_1}$ are independent. If $\bold{y} \in \mathbb{C}^{N_s}$, then the differential entropy $h(\bold{y})$ is less than or equal to $\log_2\det(\pi e \bold{Q})$ with equality if and only if $\bold{y}$ is circularly symmetric complex gaussian with $E[\bold{y}\bold{y}^H] = \bold{Q}$ \cite{Bengt}. We have
\begin{equation}\label{eq16a}
\begin{split}
E[\bold{y}\bold{y}^H] = \bold{Q} &= E \Big[ (\bold{K}\bold{x} + \bold{n_1})(\bold{K}\bold{x} + \bold{n_1})^H \Big]\\
&= E \Big[ \bold{K}\bold{x}\bold{x}^H\bold{K}^H + \bold{n_1}\bold{n_1}^H \Big]\\
&= p\bold{K}\bold{K}^H + \bold{\Phi}.
\end{split}
\end{equation}
Note that $\bold{\Phi} = {{\sigma_n^2}{\bold{G}}{\bold{G}^H} + {\bold{W}_D^H}{\bold{D}_q^2}{\bold{W}_D}}$. Thus, the differential entropies $h(\bold{y})$ and $h(\bold{n_1})$ are given below.
\begin{equation}\label{diffent}
\begin{split}
h(\bold{y}) &\le \log_2\det(\pi e \bold{Q}) = \log_2\det \bigg( \pi e \Big(p \bold{K}\bold{K}^H + \bold{\Phi} \Big) \bigg), \\
h(\bold{n1}) &\le \log_2\det(\pi e \bold{\Phi}).
\end{split}
\end{equation}
The following Theorem proves that $\bold{n_1}$ is a circularly symmetric complex Gaussian vector. 
\newtheorem{theorem}{Theorem}
\begin{theorem}
If $\bold{n}_1 = {\bold{W}_D^H}{\bold{W}_{\alpha}}{\bold{W}_A^H}{\bold{n}} + {\bold{W}_D^H}{\bold{n}_q}$, where $\bold{n}$ is  $\bold{n} \sim \mathcal{CN}(\bold{0},{\sigma_n^2\bold{I}_{N_s}})$ and ${\bold{n}_q} \sim \mathcal{N}(\bold{0},{\bold{D}_q^2})$ with ${\bold{D}_q^2} = {\bold{W}_{\alpha}}{\bold{W}_{1-\alpha}}{\text{diag}}[ {\bold{W}_A^H}{\bold{H}}({\bold{W}_A^H}{\bold{H}})^H+{\bold{I}_{N_s}}]$, then it can be shown that $\bold{n}_1$ is circularly symmetric complex Gaussian (CSCG) vector. That is, $\bold{n_1} \sim \mathcal{CN}(\bold{0},\bold{\Phi})$.
\end{theorem}
\begin{proof}
The condition for the random vector $\bold{n}_1$ to be CSCG is $\cite{Vishwa,RobertGallager}$
\begin{equation}\label{proof_eq1}
\begin{split}
&E[\bold{n}_1] = \bold{0},\\
&E[\bold{n}_1\bold{n}_1^T] = \bold{0}.
\end{split}
\end{equation}
Here, $E[\bold{n}_1\bold{n}_1^T]$ is the pseudo-covariance. 
We first prove that $\bold{n}_q$ is CSCG distributed as $\mathcal{CN}(\bold{0},{\bold{D}_q^2})$.\\
\noindent
Given ${\bold{D}_q^2} = E[\bold{n}_q\bold{n}_q^H] = {\bold{W}_{\alpha}}{\bold{W}_{1-\alpha}}{\text{diag}}[ {\bold{W}_A^H}{\bold{H}}({\bold{W}_A^H}{\bold{H}})^H+{\bold{I}_{N_s}}]$ with $\bold{W}_{\alpha}$, $\bold{W}_{1-\alpha}$,  and ${\text{diag}}[ {\bold{W}_A^H}{\bold{H}}({\bold{W}_A^H}{\bold{H}})^H+{\bold{I}_{N_s}}]$  being positive real diagonal matrices, effectively results in the covariance matrix ${\bold{D}_q^2}$ being positive real diagonal.\\
A necessary and sufficient condition  for a random vector $\bold{n}_q$ to be a circularly symmetric jointly Gaussian random vector is that it has the form $\bold{n}_q = \bold{A}\bold{w}$ where $\bold{w}$ is i.i.d. complex Gaussian, that is $\bold{w} \sim \mathcal{CN}(\bold{0},\bold{I_{N_s}})$ and $\bold{A}$ is an arbitrary complex matrix \cite{Vishwa, RobertGallager}. 
Since ${\bold{D}_q^2}$ is positive real diagonal matrix, we can express as 
\begin{equation}\label{proof_eq2}
\bold{n}_q = \bold{D}_q\bold{w},
\end{equation}
where $\bold{w} \sim \mathcal{CN}(\bold{0},\bold{I_{N_s}})$. This leads to $E[\bold{n}_q] = \bold{D}_qE[\bold{w}] = \bold{0}$ and $E[\bold{n}_q\bold{n}_q^T] = \bold{D}_qE[\bold{w}\bold{w}^T]\bold{D}_q = \bold{0}$.
Hence $\bold{n}_q$ is circularly symmetric jointly Gaussian random vector. 
\noindent
Using $\eqref{proof_eq2}$, we can express $\bold{n}_1$ as
\begin{equation}\label{proof_eq4}
\begin{split}
\bold{n}_1 = {\bold{W}_D^H}{\bold{W}_{\alpha}}{\bold{W}_A^H}{\bold{n}} + {\bold{W}_D^H}\bold{D}_q{\bold{w}}.
\end{split}
\end{equation}
Since we have $\bold{n}$ and $\bold{w}$ as i.i.d. complex Gaussian vectors, we have
\begin{equation}\label{proof_eq5}
\begin{split}
&E[\bold{n}\bold{n}^T] = E[\bold{w}\bold{n}^T] = \bold{0},\\
&E[\bold{n}\bold{w}^H] = E[\bold{w}\bold{n}^H] = \bold{0},\\
&E[\bold{n}\bold{n}^H] = \sigma_n^2\bold{I}_{N_s},\\
&E[\bold{w}\bold{w}^H] = \bold{I}_{N_s}.\\
\end{split}
\end{equation}
Thus, we arrive at 
\begin{equation}\label{proof_eq6}
\begin{split}
E[\bold{n}_1] &= {\bold{W}_D^H}{\bold{W}_{\alpha}}{\bold{W}_A^H}E[{\bold{n}}] + {\bold{W}_D^H}\bold{D}_qE[{\bold{w}}] = 0,\\
E[\bold{n}_1\bold{n}_1^T] &= \bold{G}E[\bold{n}\bold{n}^T]\bold{G}^T + \bold{G}E[\bold{n}\bold{w}^T]\bold{D}_q\bold{W}_D\\
&+ \bold{W}_D^T\bold{D}_qE[\bold{w}\bold{n}^T]\bold{G}^T + \bold{W}_D^T\bold{D}_qE[\bold{w}\bold{w}^T]\bold{D}_q\bold{W}_D\\
&= \bold{0}.
\end{split}
\end{equation}
Also, 
\begin{equation}\label{proof_eq7}
\begin{split}
E[\bold{n}_1\bold{n}_1^H] = \bold{\Phi} &= \bold{G}E[\bold{n}\bold{n}^H]\bold{G}^H + \bold{G}E[\bold{n}\bold{w}^H]\bold{D}_q\bold{W}_D\\
&+ \bold{W}_D^H\bold{D}_qE[\bold{w}\bold{n}^H]\bold{G}^H\\
& + \bold{W}_D^H\bold{D}_qE[\bold{w}\bold{w}^H]\bold{D}_q\bold{W}_D,\\
&= \sigma_n^2\bold{G}\bold{G}^H + \bold{W}_D^H\bold{D}_q^2\bold{W}_D.
\end{split}
\end{equation}
Thus, $\bold{n}_1 \sim \mathcal{CN}(\bold{0},\bold{\Phi})$ is a circularly symmetric jointly Gaussian vector.
\end{proof}
Hence, we arrive at 
\begin{equation}\label{hn_equal}
h(\bold{n1}) = \log_2\det(\pi e \bold{\Phi}).
\end{equation}
Thus, the maximum achievable mutual information $I(\bold{X};\bold{Y})$ can be written as
\begin{equation}\label{maxI}
\begin{split}
I(\bold{X};\bold{Y}) & \overset{(a)}=  h(\bold{y}) - h(\bold{n1}),\\
&= \log_2\det( \pi e \bold{Q}) - \log_2\det(\pi e \bold{\Phi}),\\
&= \log_2\det \Big (\bold{Q}\bold{\Phi}^{-1} \Big),\\
&= \log_2\det \Big ( p\bold{K}\bold{K}^H\bold{\Phi}^{-1} + \bold{I}_{N_s} \Big),
\end{split}
\end{equation}
where (a) follows from the assumption that the input symbol vector $\bold{x}$ is circular symmetric Gaussian vector that could be modeled 
as $\bold{x} \sim \mathcal{CN}(\bold{0},p\bold{I_{N_s}})$. We now simplify ($\ref{maxI}$) as
\begin{equation}\label{maxI_cont}
\begin{split}
I(\bold{X};\bold{Y}) &= \log_2\det \Big ( p\bold{K}\bold{K}^H\bold{\Phi}^{-1}\bold{K}\bold{K}^{-1} + \bold{K}\bold{K}^{-1} \Big),\\
&= \log_2\det \Big ( p\bold{K} \big ( \bold{K}^H\bold{\Phi}^{-1}\bold{K} + \frac{1}{p}\bold{I}_{N_s} \big) \bold{K}^{-1} \Big),\\
&= \log_2\det ( p\bold{K}) \det \Big ( \bold{K}^H\bold{\Phi}^{-1}\bold{K} + \frac{1}{p}\bold{I}_{N_s} \Big) \det (\bold{K}^{-1}),\\
&= \log_2 p^{N_s}\det \Big ( \bold{K}^H\bold{\Phi}^{-1}\bold{K} + \frac{1}{p}\bold{I}_{N_s} \Big).
\end{split}
\end{equation}
The capacity is computed by maximizing \eqref{maxI_cont}. We are maximizing $I(\bold{X};\bold{Y})$ for a given fixed channel $\bold{H}$, and for a given combiner set $\bold{\tilde{W}}_A^H$,$\bold{W}_D^H$. Hence, maximization of \eqref{maxI_cont} will be over the bit-allocation matrix $\bold{W}_{\alpha}$. $\bold{\Phi}$ is a function of $\bold{W}_{\alpha}$ and $\bold{K}$ is constant for fixed $\bold{H}$, $\bold{\tilde{W}}_A^H$, and $\bold{W}_D^H$.
\begin{equation}\label{maxI_cont2}
\begin{split}
C &= \max \Bigg\{ \log_2 p^{N_s}\det \Big ( \bold{K}^H\bold{\Phi}^{-1}\bold{K} + \frac{1}{p}\bold{I}_{N_s} \Big) \Bigg\} \\
 &= {N_s}\log_2 p + \log_2\det \Big ( ({\bold{I}^{-1}({\bold{\hat{x}}})})^{-1} + \frac{1}{p}\bold{I}_{N_s} \Big).
\end{split}
\end{equation}
Note that ${\bold{I}^{-1}({\bold{\hat{x}}})}$ is the Cramer-Rao Lower Bound (CRLB) that can be achieved by an efficient estimator (if one exists) for estimating the transmitted symbol vector $\bold{x}$, given the observations $\bold{y}$ in ($\ref{eq9a}$). In fact, it can be shown that for ($\ref{eq9a}$), there exists an efficient estimator such that 
$\bold{I}^{-1}({\bold{\hat{x}}}) = (\bold{K}^H\bold{\Phi}^{-1}\bold{K})^{-1}$ \cite{Zakir2}.
\subsection{Condition for optimal bit allocation}
Using the capacity expression derived in ($\ref{maxI_cont2}$), we now maximize the capacity by deriving a condition for an optimal bit allocation. The capacity is maximized by selecting some $\bold{b}^*$ that satisfies the ADC power constraint. We can write this expression for the maximum capacity from ($\ref{maxI_cont2}$) as 
\begin{equation}\label{maxcap}
\begin{split}
C &= {N_s}\log_2p \text{ }+ \\
&\underbrace{\text{max}}_{\substack{\bold{b}^*; {P_{\text{TOT}}}\leq{P_{\text{ADC}}}}}\Bigg\{ \log_2  \det \Big ( ({\bold{I}^{-1}({\bold{\hat{x}}})})^{-1} + \frac{1}{p}\bold{I}_{N_s} \Big) \Bigg\}.
\end{split}
\end{equation}
The condition for optimal ADC bit allocation $\bold{b}^*$ that optimizes ($\ref{maxcap}$) is given by
\begin{equation}\label{bitcond}
\bold{b}^* = \underbrace{\text{argmax}}_{\substack{\bold{b} \in \mathbb{I}^{N_s \times 1}; \\ {P_{\text{TOT}}}\leq{P_{\text{ADC}}}}}\Bigg\{ \log_2  \det \Big ( ({\bold{I}^{-1}({\bold{\hat{x}}})})^{-1} + \frac{1}{p}\bold{I}_{N_s} \Big) \Bigg\}.
\end{equation}
The CRLB for the linear estimator in ($\ref{eq9a}$) is derived  \cite{Zakir2} as
\begin{equation}\label{crlb}
\begin{split}
{\bold{I}^{-1}({\bold{\hat{x}}})} &= ({\bold{K}^H}{\bold{\Phi}^{-1}}{\bold{K}})^{-1}\\
&={\sigma_n^2}{\bold{\Sigma}^{-2}}+{\bold{K}^{-1}}{\bold{W}_D^H}{\bold{D}_q^2}{\bold{W}_D}({\bold{K}^H})^{-1}.
\end{split}
\end{equation}
By substituting $\bold{K}$ into ($\ref{crlb}$) and designing the analog combiner $\bold{\tilde{W}}_A^H$ adhering to the constraints imposed by the phase shifters and digital combiner as ${\bold{U}^H} = {\bold{W}_D}{\bold{\tilde{W}}_A^H}$ defined in \cite{Zakir2}, we have
\begin{equation}\label{crlb_cont}
\begin{split}
{\bold{I}^{-1}({\bold{\hat{x}}})} &= {\sigma_n^2}{\bold{\Sigma}^{-2}}\\
&+{\bold{\Sigma}}^{-1}{\bold{U}^H}({\bold{W}_A^H})^{-1}{\bold{W}_{\alpha}^{-1}}{\bold{D}_q^2}{\bold{W}_{\alpha}^{-1}}{\bold{W}_A^{-1}}{\bold{U}}{\bold{\Sigma}}^{-1},\\
&= {\sigma_n^2}{\bold{\Sigma}^{-2}}+{\bold{\Sigma}}^{-2}{\bold{W}_{\alpha}^{-2}}{\bold{D}_q^2}.
\end{split}
\end{equation}
We now compute the Inverse of CRLB $\Big({\bold{I}^{-1}({\bold{\hat{x}}})}\Big)^{-1}$ as
\begin{equation}\label{inv_crlb}
\begin{split}
&\Big({\bold{I}^{-1}({\bold{\hat{x}}})}\Big)^{-1} = \Big({\sigma_n^2}{\bold{\Sigma}^{-2}}+{\bold{\Sigma}}^{-2}{\bold{W}_{\alpha}^{-2}}{\bold{D}_q^2}\Big)^{-1},\\
&= \diagBigg{\frac{\sigma_1^2}{\sigma_n^2 + \frac{f(b_1)l_1}{\big(1-f(b_1)\big)}}, \cdots, \frac{\sigma_{N_s}^2}{\sigma_n^2 + \frac{f(b_{N_s})l_{N_s}}{\big(1-f(b_{N_s})\big)}}}.
\end{split}
\end{equation}
Note that $\bold{W}_{\alpha}(\bold{b}) = \diag{1-f(b_1), \cdots, 1-f(b_{N_s})}$, $\bold{D}_q^2 = \diag{(1-f(b_1))f(b_1)l_1, \cdots, (1-f(b_{N_s}))f(b_{N_s})l_{N_s}}$, and $l_i = ({1+{\bold{w}_{A_i}^H}{\bold{h_i}^H}{\bold{h_i}}{\bold{w}_{A_i}}})$. Further, ${\sigma_i}$ is the diagonal element of ${\bold{\Sigma}}$, $f(b_i)$ is the ratio of the Mean Square Quantization Error (MQSE) and the power of the symbol for a non-uniform MMSE quantizer with $b_i$ bits along the RF path $i$, $i=1,2,\ldots N_s$ $\cite{VarBitAlloc}$. The values for $f(b_i)$ are indicated in  Table $\ref{betaVal}$. We set $\bold{w}_{A_i}$ and $\bold{h}_i$ as the $i^{th}$ columns of the matrix $\bold{W_A}$ and $\bold{H}^H$, respectively. Substituting ($\ref{inv_crlb}$) in ($\ref{bitcond}$), we have
\begin{equation}
\begin{split}
\bold{b}^* &= \underbrace{\text{argmax}}_{\substack{\bold{b} \in \mathbb{I}^{N_s \times 1}; \\ {P_{\text{TOT}}}\leq{P_{\text{ADC}}}}} \log_2 \det \diagBigg{\frac{\sigma_1^2}{\sigma_n^2 + \frac{f(b_1)l_1}{\big(1-f(b_1)\big)}} + \frac{1}{p}}, \nonumber
\end{split}
\end{equation}
\begin{equation}\label{bitcond_cont}
\begin{split}
&= \underbrace{\text{argmax}}_{\substack{\bold{b} \in \mathbb{I}^{N_s \times 1}; \\ {P_{\text{TOT}}}\leq{P_{\text{ADC}}}}} \log_2 \prod_{i=1}^{N_s} \Bigg( \frac{\sigma_i^2}{\sigma_n^2 + \frac{f(b_i)l_i}{\big(1-f(b_i)\big)}} + \frac{1}{p} \Bigg),\\
&= \underbrace{\text{argmax}}_{\substack{\bold{b} \in \mathbb{I}^{N_s \times 1}; \\ {P_{\text{TOT}}}\leq{P_{\text{ADC}}}}} \sum_{i=1}^{N_s} \Bigg\{ \log_2  \Bigg( q(b_i) + 1 \Bigg) \Bigg\},
\end{split}
\end{equation}
where $q(b_i) = \frac{p\sigma_i^2}{\sigma_n^2 + \frac{f(b_i)l_i}{\big(1-f(b_i)\big)}}$. The term $\log_2  \bigg( q(b_i) + 1 \bigg)$ can be expanded using series expansion for two scenarios. In the first case, we expand the term for $0 \leq q(b_i) < 1$. In second case, we expand the term using Taylor series for $1  \leq q(b_i) < \infty$. Due to page limitation, the proofs for (\ref{px_less_1}) and (\ref{log2_taylor}) are omitted.
\subsubsection{Case-1}
The term $\log_2  \bigg( q(b_i) + 1 \bigg)$ for $0 \leq q(b_i) < 1$, can be written as: 
\begin{equation}\label{px_less_1}
\log_2  \bigg( q(b_i) + 1 \bigg) \simeq \frac{q(b_i)}{\ln2}. 
\end{equation}
Thus the maximization in $\eqref{bitcond_cont}$ can be written as
\begin{equation}\label{bitcond_cont2}
\bold{b}^* = \underbrace{\text{argmax}}_{\substack{\bold{b} \in \mathbb{I}^{N_s \times 1}; \\ {P_{\text{TOT}}}\leq{P_{\text{ADC}}}}} \sum_{i=1}^{N_s} \frac{p\sigma_i^2}{\sigma_n^2 + \frac{f(b_i)l_i}{\big(1-f(b_i)\big)}}.
\end{equation}
\subsubsection{Case-2}
The term $\log_2  \bigg( q(b_i) + 1 \bigg)$ can be expanded using the Taylor series for $\infty > q(b_i) \geq 1$ as 
\begin{equation}\label{log2_taylor}
\log_2  \bigg( q(b_i) + 1 \bigg) = \Bigg(1-\frac{1}{q(b_i)}\Bigg)P + L(p,\sigma_i^2, \sigma_n^2)
\end{equation}
where $P$ and $L(p,\sigma_i^2, \sigma_n^2)$ are terms that are not a function of $b_i$. The maximization in $\eqref{bitcond_cont}$ can be simplified to 
\begin{equation}\label{bitcond_cont1}
\begin{split}
\bold{b}^* &= \underbrace{\text{argmax}}_{\substack{\bold{b} \in \mathbb{I}^{N_s \times 1}; \\ {P_{\text{TOT}}}\leq{P_{\text{ADC}}}}} \sum_{i=1}^{N_s} \Bigg(1-\frac{1}{q(b_i)}\Bigg), \\
&= \underbrace{\text{argmax}}_{\substack{\bold{b} \in \mathbb{I}^{N_s \times 1}; \\ {P_{\text{TOT}}}\leq{P_{\text{ADC}}}}} \sum_{i=1}^{N_s} \frac{p\sigma_i^2}{\sigma_n^2 + \frac{f(b_i)l_i}{\big(1-f(b_i)\big)}}.
\end{split}
\end{equation}
We observe that $\eqref{bitcond_cont1}$ is the same as $\eqref{bitcond_cont2}$. Hence both scenarios lead to the same optimization problem. We now define the term $K_f(b_i)$ for a given bit allocation $b_i$ on RF path $i$ as 
\begin{equation}\label{gain_term}
k_f(b_i) \triangleq \frac{p\sigma_i^2}{\sigma_n^2 + \frac{f(b_i)l_i}{\big(1-f(b_i)\big)}}.
\end{equation}
For a given bit allocation $\bold{b}_j$ in $B_{\text{set}}$, we define
\begin{equation}\label{gain_termsum}
K_f(\bold{b}_j) \triangleq \sum_{i=1}^{N_s} \frac{p\sigma_i^2}{\sigma_n^2 + \frac{f(b_i)l_i}{\big(1-f(b_i)\big)}},
\end{equation}
where $B_{\text{set}}$ is a set of all possible bit allocations that satisfy a given ADC power budget $P_{ADC}$ for given $N_s$ \cite{Zakir2, Zakir1}.
\begin{equation}\label{bset}
\begin{split}
B_{\text{set}} \triangleq \big\{ &\bold{b}_j = {\big[ b_{j1}, b_{j2}, \dots, b_{jN}  \big]}^T \text{ for } 0 \leq j < 4^{N_s} \mid \\
& 1 \le b_{ji} \le 4 \text{ and } \sum_{i=1}^{N} cf_s2^{b_{ji}} \leq P_{\text{ADC}} \big\}.
\end{split}
\end{equation}
The maximization in $\eqref{bitcond_cont1}$ can be written as
\begin{equation}\label{bitcond_cont3}
\bold{b}^* = \underbrace{\text{argmax}}_{\substack{\bold{b}_j \in B_{\text{set}}; \\ {P_{\text{TOT}}}\leq{P_{\text{ADC}}}}} K_f(\bold{b}_j).
\end{equation}
\indent
Interestingly, this is the same condition that was derived for optimal bit allocation by minimizing the MSE criterion in \cite{Zakir2}.\\
\indent
It is also worth noting that using a high-resolution ADCs, $\bold{D}_q^2 = \bold{0}$ and the CRLB defined in ($\ref{crlb}$) reduces to ${\bold{I}^{-1}({\bold{\hat{x}}})} = {\sigma_n^2}{\bold{\Sigma}^{-2}}$. Substituting the same into ($\ref{maxI_cont2}$), we can write the expression for the capacity for the given channel with infinite-resolution ADCs as
\begin{equation}\label{infcap}
\begin{split}
C_{\infty} &= \log_2 p^{N_s}\det \Big ( \frac{1}{\sigma_n^2}\bold{\Sigma}^2 + \frac{1}{p}\bold{I}_{N_s} \Big),\\
&= \log_2 \det \Big ( \frac{p}{\sigma_n^2}\bold{\Sigma}^2 + \bold{I}_{N_s} \Big).
\end{split}
\end{equation}
With uniform power allocation on the transmitter, $p$ is uniformly divided along $N_s$ RF paths and the capacity with uniform power allocation at the transmitter becomes
\begin{equation}\label{infcap_unfm}
C_{\infty} = \sum_{i=1}^{N_s} {\log_2 \Bigg( \frac{\rho}{N_s}\sigma_i^2 + 1\Bigg)},
\end{equation}
where $\rho = \frac{p}{\sigma_n^2}$ is the average SNR at the receiver.\\
Similarly, with perfect Channel State Information at the Transmitter and waterfilling, the capacity with high resolution ADCs can be written as
\begin{equation}\label{infcap_waterfill}
C_{\infty} = \sum_{i=1}^{N_s} {\log_2 \Bigg( \epsilon_i\frac{\rho}{N_s}\sigma_i^2 + 1\Bigg)},
\end{equation}
where $\epsilon_i$ is the portion of the total power $p$ allocated to RF path $i$ at the transmitter based on water-filling algorithm \cite{Vishwa}. Thus ($\ref{infcap_unfm}$) and ($\ref{infcap_waterfill}$) derived in \cite{Bengt} are special cases of (\ref{maxI_cont2}).
\begin{table}[t]
\begin{tabu} to 0.5\textwidth { c c c c c c }
\hline
$b_i$  & 1 & 2 & 3 & 4 & 5 \\
\hline
$f(b_i)$ & 0.3634 & 0.1175 & 0.03454 & 0.009497 & 0.002499 \\
\hline
\end{tabu}
\vspace{1mm}
\caption{Values of $f(b_i)$ for different ADC Quantization Bits $b_i$} \label{betaVal}
\end{table}
\section{Simulations} \label{Sim}
We simulate the mmWave channel using the NYUSIM channel simulator with the configurations specified in Table $\ref{nyusimtab}$ $\cite{nyusim}$. We strengthen the singular value on the dominant channel to simulate a strong scatterer $\cite{SVDcorrChannels}$. We consider $N_s=8$ or $N_s=12$ strong channels (RF paths) for capacity simulations. The combiners are designed as per ($\ref{CombDesign}$).\\
\indent
We run the simulations to evaluate the capacity as derived in ($\ref{maxI_cont2}$). The plots obtained at various SNRs are shown for $N_s = 8$ and $N_s = 12$ in Figure $\ref{fig:crlb_Capacity_Nr8_H64by32.eps}$ and Figure $\ref{fig:crlb_Capacity_Nr12_H64by32.eps}$, respectively. The simulations for capacity are obtained with all 1-bit ADCs, 2-Bit ADCs and with no quantization across RF paths. This is indicated in Figure $\ref{fig:crlb_Capacity_Nr8_H64by32.eps}$ and Figure $\ref{fig:crlb_Capacity_Nr12_H64by32.eps}$ using lines (a), (b) and (d) respectively. We also evaluate the capacity ($\ref{maxI_cont2}$) with all possible bit configurations shown in ($\ref{bset}$) that satisfy a given ADC power budget. We then pick the bit configuration, that results in the maximum capacity and call it as the Exhaustive Search (ES) solution. We also evaluate the bit configuration that maximizes the capacity based on our proposed condition in ($\ref{bitcond_cont3}$). The algorithm for arriving at this condition is similar to the one that minimizes the MSE in \cite{Zakir2}. We notice that the capacity obtained with the bit configuration solution from our proposed approach (line-(e)) is very close to the exhaustive search solution (line-(c)). \\
\indent
The condition for optimal bit-allocation based on capacity maximization \eqref{bitcond_cont3} is the same as the condition for optimal bit allocation based on CRLB minimization derived in \cite{Zakir2}. Hence, the bit-allocation algorithm are same in both cases. It is shown that the computational complexity of the bit-allocation algorithm in \cite{Zakir2} has an order of magnitude improvement over Exhaustive Search (ES) technique and Genetic Algorithm (GA) technique based in \cite{Zakir1}. Table \ref{tab:CRLBTab1} summarizes the computational advantage of the proposed bit-allocation algorithm over ES and GA \cite{Zakir2}.
\begin{table}
\begin{center}
\begin{tabu} to 0.5\textwidth { | p{1cm} | p{1cm} | X[c] | X[c] | p{1cm} | X[c] | p{0.8cm} |}
\hline
\multirow{3}{1.5cm}{Number of RF paths} & \multicolumn{3}{c|}{Number of complex} & %
    \multicolumn{3}{c|}{Number of complex}\\
& \multicolumn{3}{c|}{multiplications$^{\S}$} & \multicolumn{3}{c|}{additions$^{\S}$}  \\
\cline{2-7}
& \centering ES & GA & \scriptsize Proposed-Algo. & \centering ES & GA & \scriptsize Proposed-Algo.\\
\hline
\centering \multirow{2}{*}{8}  & \scriptsize1622592 & \scriptsize 279936  & \scriptsize \textcolor{red}{864} & \scriptsize1592544 & \scriptsize274752 & {\scriptsize \textcolor{red}{760}} \\
& & & & & & {\scriptsize \textcolor{red}{13146$^{\dagger}$}}\\
\hline
\centering \multirow{2}{*}{12} & \scriptsize179092032 & \scriptsize 2721600 & \scriptsize \textcolor{red}{1296} & \scriptsize175893960 & \scriptsize 2673000 & \scriptsize \textcolor{red}{1140} \\
& & & & & & {\scriptsize \textcolor{red}{1465783$^{\dagger}$}}\\
\hline
\end{tabu}
\footnotesize{$^{\dagger}$ Real additions} \\
\vspace{0.5mm}
\caption{Computational complexity in terms of total number of multiplications and additions} \label{tab:CRLBTab1}
\end{center}
\end{table}
\begin{figure}[t]
\centering
\includegraphics[width=0.4\textwidth,scale=0.4]{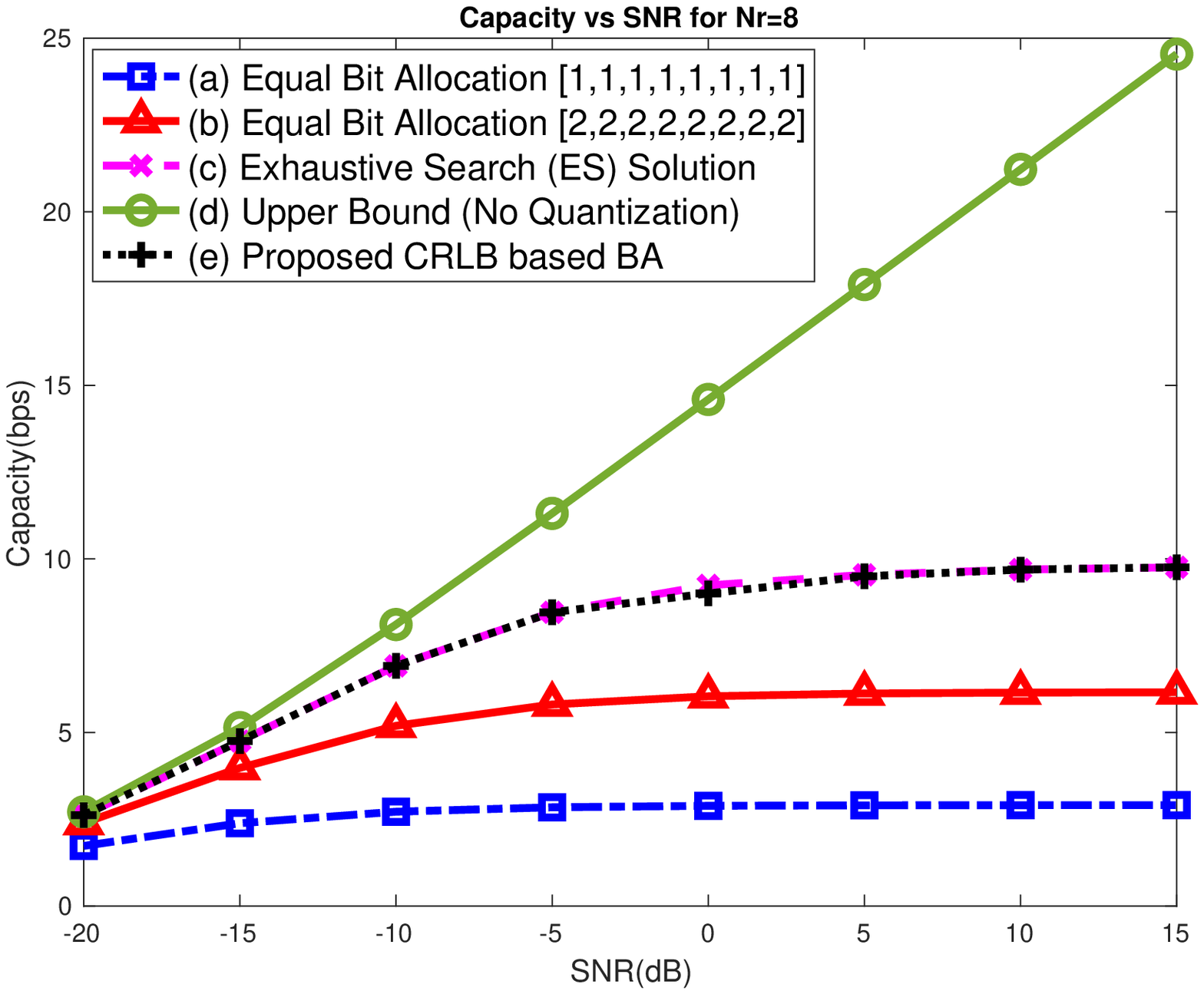}
\vspace{-0.2in}
\caption{Capacity vs. SNR for $N_s=8$ for all 1-bit, 2-bit, ES and Proposed method bit configurations}
\label{fig:crlb_Capacity_Nr8_H64by32.eps}
\includegraphics[width=0.4\textwidth,scale=0.4]{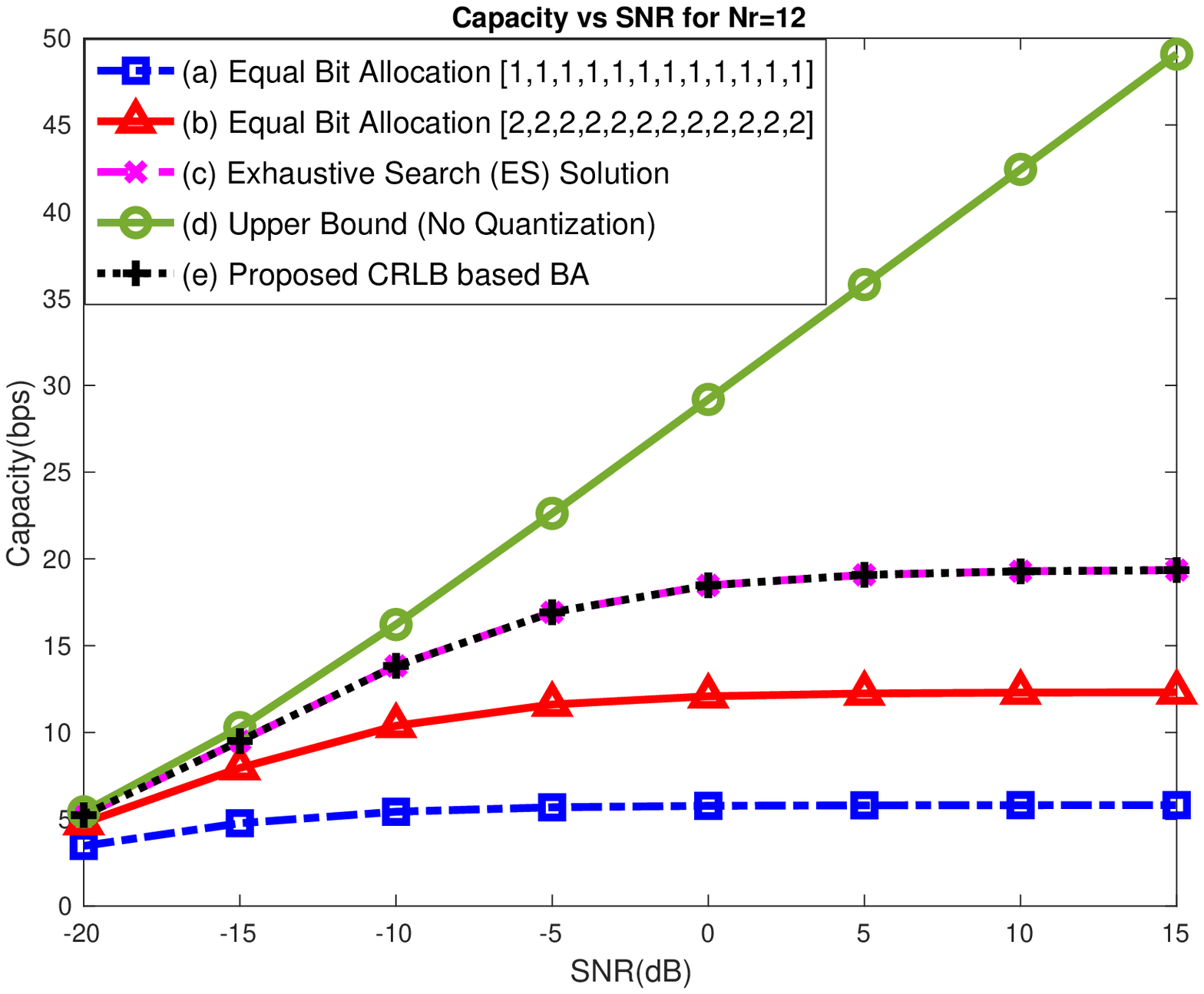}
\vspace{-0.2in}
\caption{Capacity vs. SNR for $N_s=12$ for all 1-bit, 2-bit, ES and Proposed method bit configurations}
\label{fig:crlb_Capacity_Nr12_H64by32.eps}
\end{figure}
\begin{table}
\begin{center}
\begin{tabu} to 0.5\textwidth {| l| l| }
 \hline
 \textbf{Parameters}  & \textbf{Value/Type} \\
 \hline
Frequency & 28Ghz \\
\hline
Environment & Line of sight \\
\hline
T-R seperation & 100m\\
\hline
TX/RX array type & ULA\\
\hline 
Num of TX/RX elements $N_t$/$N_r$ & 32/64\\
\hline
TX/RX  antenna spacing & $\lambda/2$\\
\hline
\end{tabu}
\vspace{1mm}
\caption{$\text{Channel parameters for NYUSIM model $\cite{nyusim}$}$} \label{nyusimtab}
\vspace{-3mm}
\end{center}
\end{table}
\section{Conclusion}
In this paper, we derive an expression for the capacity of a mmWave massive MIMO system for a given channel, with a receiver comprising of variable-bit resolution ADCs and hybrid combiner. We show that the capacity is a function of CRLB, which in turn is a function of bit-allocation matrix and hybrid combining matrices. The MSE for the quantized and combined vector achieves CRLB \cite{Zakir2}. We show that by maximizing the expression for the capacity derived for a given channel and hybrid combiner, we arrive at a condition for optimal bit-allocation. We also show that this condition is same as the one that is obtained by minimizing the CRLB in \cite{Zakir2}. 

We support the above claims through simulations. It is seen that the capacity evaluated at various SNRs using proposed bit-allocation design is very close to the Exhaustive Search (ES) technique. Also, we see that for the channel conditions stated in Section-\ref{Sim}, the optimal bit-allocation is not all equal bits on the receiver's RF paths for given ADC power budget. The proposed bit-allocation algorithm has the same computational complexity as that of \cite{Zakir2}, which has an order of magnitude improvement over ES and GA technique based on \cite{Zakir1}.

\section*{Acknowledgment}
The authors would like to thank National Instruments for the support extended for this work.  



%
\bibliographystyle{IEEEtran}
\bibliography{Pap3BibTexFile}

\begin{thebibliography}{10}
\providecommand{\url}[1]{#1}
\csname url@samestyle\endcsname
\providecommand{\newblock}{\relax}
\providecommand{\bibinfo}[2]{#2}
\providecommand{\BIBentrySTDinterwordspacing}{\spaceskip=0pt\relax}
\providecommand{\BIBentryALTinterwordstretchfactor}{4}
\providecommand{\BIBentryALTinterwordspacing}{\spaceskip=\fontdimen2\font plus
\BIBentryALTinterwordstretchfactor\fontdimen3\font minus
  \fontdimen4\font\relax}
\providecommand{\BIBforeignlanguage}[2]{{%
\expandafter\ifx\csname l@#1\endcsname\relax
\typeout{** WARNING: IEEEtran.bst: No hyphenation pattern has been}%
\typeout{** loaded for the language `#1'. Using the pattern for}%
\typeout{** the default language instead.}%
\else
\language=\csname l@#1\endcsname
\fi
#2}}
\providecommand{\BIBdecl}{\relax}
\BIBdecl

\bibitem{5GBackHaul}
{Zhen Gao, Linglong Dai, De Mi Zhaocheng Wang, Muhammed Ali Imran, and Muhammed
  Zeeshan Shakir}, ``{mm}{W}ave massive-mimo-based wireless backhaul for the
  {5G} ultra-dense network,'' \emph{IEEE Wireless Comun.}, 2015.

\bibitem{SigProc}
{Robert. W. Heath Jr., Nuria Gonzalez-Prelcic, Sundeep Rangan, Wonil Roh, Akbar
  M. Sayeed}, ``{An overview of signal processing techniques for millimeter
  wave MIMO systems},'' \emph{IEEE Journ. of Selected Topics in Signal
  Processing}, vol.~10, no.~3, 2016.

\bibitem{mmPreCom}
{Ahmed Alkhateeb, Jianhua Mo, Nuria González-Prelcic, and Robert W. Heath Jr},
  ``{MIMO precoding and combining solutions for millimeter-wave systems},''
  \emph{IEEE Comun. Magazine}, 2014.

\bibitem{Rangan}
{Oner Orhan, Elza Erkip, and Sundeep Rangan}, ``{Low power analog to- digital
  conversion in millimeter wave systems: Impact of resolution and bandwidth on
  performance},'' \emph{Proc. IEEE Info. Theory and Applications Workshop}, pp.
  191--198, Feb. 2015.

\bibitem{Muris}
{Muris Sarajlic, Liang Liu and Ove Edfors}, ``{When Are Low Resolution ADCs
  Energy Efficient in Massive MIMO?}'' \emph{IEEE Access}, vol.~5, pp. 14\,837
  -- 14\,853, 2017.

\bibitem{HybArchCap}
{Jianhua Mo, Ahmed Alkhateeb, Shadi Abu-Surra and Robert W. Heath}, ``{Hybrid
  Architectures With Few-Bit ADC Receivers: Achievable Rates and Energy-Rate
  Tradeoffs},'' \emph{IEEE Tran. on Wireless Communications}, vol.~16, no.~4,
  pp. 2274--2287, Oct. 2017.

\bibitem{Mezghani}
{Amine Mezghani and Josef A. Nossek}, ``{On ultra-wideband MIMO systems with
  1-bit quantized outputs: Performance analysis and input optimization},''
  \emph{IEEE Int. Symp. Inf. Theory}, p. 1286?1289, 2007.

\bibitem{Jmo}
J.~Mo and R.~W.~H. Jr., ``{Capacity analysis of one-bit quantized MIMO systems
  with transmitter channel state information},'' \emph{IEEE Tran. on Signal
  Processing}, vol.~63, no.~20, p. 1286?1289.

\bibitem{Uplink}
{Li Fan, Shi Jin, Chao-Kai Wen, and Haixia Zhang}, ``{Uplink achievable rate
  for massive MIMO systems with low resolution ADC},'' \emph{IEEE Comun.
  Letters}, vol.~19, no.~12, pp. 2186--2189, Oct. 2015.

\bibitem{Zakir1}
\BIBentryALTinterwordspacing
{I. Zakir Ahmed, Hamid Sadjadpour, Shahram Yousefi}, ``{A joint combiner and
  bit allocation design for massive MIMO using genetic algorithm},''
  \emph{Proc. of Asilomar Conf. on Signals, Systems and Computers 2017}.
  [Online]. Available: \url{https://export.arxiv.org/pdf/1711.06706}
\BIBentrySTDinterwordspacing

\bibitem{VarBitAlloc}
\BIBentryALTinterwordspacing
{Jinseok Choi and Brian L. Evans, Alan Gatherer}, ``{ADC bit Allocation under a
  power constraint for mmWave massive MIMO communication receivers}.''
  [Online]. Available: \url{https://arxiv.org/abs/1609.05165}
\BIBentrySTDinterwordspacing

\bibitem{Zakir2}
\BIBentryALTinterwordspacing
{I. Zakir Ahmed, Hamid Sadjadpour, Shahram Yousefi}, ``{Single-user mmwave
  massive MIMO: SVD-based adc bit allocation and combiner design},''
  \emph{Proc. of Int. Conf. on Signal Processing and Communications 2018}.
  [Online]. Available: \url{https://arxiv.org/pdf/1804.08595.pdf}
\BIBentrySTDinterwordspacing

\bibitem{rapaport}
{Theodore S. Rappaport, Robert W. Heath Jr, Robert C. Daniels, James N.
  Murdock}, \emph{{Millimeter Wave Wireless Communications}}.\hskip 1em plus
  0.5em minus 0.4em\relax Prentice Hall Press, 2015.

\bibitem{PreDsgn}
{Omar El Ayach, Sridhar Rajagopal, Shadi Abu-Surra, Zhouyue Pi, Robert. W.
  Heath Jr.}, ``{Spatially sparse precoding in millimeter wave MIMO systems},''
  \emph{IEEE Journ. in Selected Areas of Comm.}, vol.~8, no.~3, 2017.

\bibitem{Thomas}
{Thomas M. Cover, Joy A. Thomas}, ``{Elements of Information Theory},''
  \emph{John Wiley and Sons}, 1991.

\bibitem{Bengt}
{Bengt Holter}, ``{On the Capacity of the MIMO Channel - A Tutorial
  Introduction}.''

\bibitem{Vishwa}
{David Tse, Pramod Viswanath}, ``{Fundamentals of Wireless Communication},''
  \emph{Cambridge University Press}, 2005.

\bibitem{RobertGallager}
{Robert G. Gallager}, ``{Stochastic Processes: Theory for Applications},''
  \emph{Cambridge University Press}, 2013.

\bibitem{nyusim}
{Shu Sun, George R. MacCartney Jr.,Theodore S. Rappaport}, ``{A Novel
  millimeter-wave channel simulator and applications for 5G wireless
  communications},'' \emph{2017 IEEE Int. Conf. on Comun. (ICC)}, 2007.

\bibitem{SVDcorrChannels}
{David W. Browne, Michael W. Browne, Michael P. Fitz}, ``{Singular Value
  Decomposition of correlated MIMO channels},'' \emph{Proc. of the IEEE
  GLOBECOM}, 2006.

\end{thebibliography}
\end{document}